\documentclass[12pt, english]{article}

\usepackage[utf8]{inputenc}
\usepackage{amsmath}
\usepackage{amsthm}
\usepackage{amstext}
\usepackage{amsmath}
\newcommand{\footremember}[2]{%
    \footnote{#2}
    \newcounter{#1}
    \setcounter{#1}{\value{footnote}}%
}
\usepackage{array}
\usepackage{tabu}
\usepackage[T1]{fontenc}
\usepackage{babel}
\usepackage{array}
\usepackage{amssymb}
\usepackage{amsfonts}
\usepackage{xspace}

\usepackage{amsthm}
\newtheorem{theorem}{Theorem}[section]
\newtheorem{definition}{Definition}[section]

\newtheorem{corollary}[theorem]{Corollary}
\newtheorem{lemma}[theorem]{Lemma}

\newcommand{\cAkNN}{$(c,k)$-NN\xspace}
\renewcommand{\epsilon}{\varepsilon}

\title{On the I/O complexity of the k-nearest neighbor problem}

\author{Mayank Goswami\footremember{mayank}{Queens College CUNY, Flushing, New York, USA. Supported by NSF grants CRII-1755791 and CCF-1910873. \texttt{mayank.goswami@qc.cuny.edu}}
\and Riko Jacob\footremember{riko}{IT University of Copenhagen, København S, Denmark. Part of this work done during the second Hawaiian workshop on parallel algorithms and data structures, University of Hawaii at Manoa, Hawaii, USA. \texttt{rikj@itu.dk}}\\
\and Rasmus Pagh\footremember{pagh}{BARC and IT University of Copenhagen, København S, Denmark. Work done in part while visiting Google Research. Supported by Investigator Grant 16582, Basic Algorithms Research Copenhagen (BARC), from the VILLUM Foundation, and by funding from the European Research Council under the European Union’s 7th Framework Programme (FP7/2007-2013) / ERC grant agreement no.~614331. \texttt{pagh@itu.dk}}}

\date{}
\begin{document}

\maketitle

% \author{Mayank Goswami}
% \email{mayank.goswami@qc.cuny.edu}
% \affiliation{  \institution{Queens College CUNY}
%   \city{Flushing}
%   \state{NY}
%   \postcode{11367}
%   \country{USA}
% }

% \author{Riko Jacob}
% \email{rikj@itu.dk}
% \orcid{0000-0001-9470-1809}
% \affiliation{  \institution{IT University of Copenhagen}
%   \city{København S}
%   \state{}
%   \postcode{2300}
%   \country{Denmark}
% }

% \author{Rasmus Pagh}
% \authornote{Work done in part while visiting Google Research. Supported by Investigator Grant 16582, Basic Algorithms Research Copenhagen (BARC), from the VILLUM Foundation, and by funding from the European Research Council under the European Union’s 7th Framework Programme (FP7/2007-2013) / ERC grant agreement no.~614331.}
% \email{pagh@itu.dk}
% \orcid{0000-0002-1516-9306}
% \affiliation{  \institution{BARC and IT U.~Copenhagen}
%   \city{København S}
%   \state{}
%   \postcode{2300}
%   \country{Denmark}
% }

%\keywords{Nearest neighbors, I/O complexity, indexability model}

\begin{abstract}

We consider static, external memory indexes for exact and approximate versions of the $k$-nearest neighbor ($k$-NN) problem, and show new lower bounds under a standard \emph{indivisibility} assumption:

\begin{itemize}
    \item Polynomial space indexing schemes for high-dimensional $k$-NN in Hamming space cannot take advantage of block transfers: $\Omega(k)$ block reads are needed to to answer a query.
    \item For the $\ell_\infty$ metric the lower bound holds even if we allow $c$-appoximate nearest neighbors to be returned, for $c \in (1, 3)$.		

    \item The restriction to $c < 3$ is necessary: For every metric there exists an indexing scheme in the \emph{indexability model} of Hellerstein et al.~using space $O(kn)$, where $n$ is the number of points, that can retrieve $k$ 3-approximate nearest neighbors using {optimal} $\lceil k/B\rceil$ I/Os, where $B$ is the block size.
    \item For specific metrics, data structures with better approximation factors are possible.
		For $k$-NN in Hamming space and every approximation factor $c>1$
    there exists a polynomial space data structure that returns $k$ $c$-approximate nearest neighbors in $\lceil k/B\rceil$ I/Os.
\end{itemize}

To show these lower bounds we develop two new techniques:
First, to handle that approximation algorithms have more freedom in deciding which result set to return we develop a \emph{relaxed} version of the $\lambda$-set workload technique of Hellerstein et al.
This technique allows us to show lower bounds that hold in $d\geq n$ dimensions.
To extend the lower bounds down to $d = O(k \log(n/k))$ dimensions, we develop a new deterministic dimension reduction technique that may be of independent interest.
 \end{abstract}

\section{Introduction}

The DEEP1B data set~\cite{babenko2016efficient} is among the largest image data sets that has been examined in the similarity search literature.
From each of $n = 10^9$ images, a 96-dimensional vector has been extracted from an intermediate layer of a pre-trained deep neural network, a state-of-the-art method for semantically meaningful feature vectors~\cite{babenko2014neural}.
Such feature vectors can be thought of as compressed representations of the images that, for example, can be used to estimate the similarity of two images.
In many use cases, though, it is not enough to substitute the images with their feature vectors, but we also need to be able to access the corresponding images.
Though the size of the raw image data behind DEEP1B is not stated in~\cite{babenko2016efficient}, an estimate would be 1 MB per image on average, or 1000 TB in total.
Clearly, retrieving similar images in a data set of this size is beyond what is possible on a single machine, and even just indexing the set of feature vectors would require an amount of internal memory that is larger than what is present in most servers.

\paragraph{$k$-nearest neighbors}

In the $k$-nearest neighbors ($k$-NN) problem we want to construct a data structure on a set $P$ of $n$ points in some metric space that, given an integer $k > 0$ and a query point $q$, finds the $k$ closest points to $q$ in $P$.
We will be focusing on data structures that can be constructed in polynomial time and space.
The $k$-NN problem is believed to be hard in high dimensions even for $k=1$, and the brute-force algorithm that considers all data points in $P$ essentially optimal.
In particular, Williams (see~\cite{Williams}) proved that for constant $\epsilon > 0$, no $O(n^{1-\epsilon})$ query time data structure is possible for $\omega(\log n)$-dimensional Hamming space, assuming the Strong Exponential Time Hypothesis.

Because of the hardness of the problem most research has revolved around \emph{approximate} solutions. 
The \emph{$c$-approximate} $k$-NN (\cAkNN) problem asks to return $k$ points from $P$ with distance at most $cr$ to $q$, where $r$ is the distance to the $k$th nearest neighbor of $q$.
It is known that \cAkNN is equivalent, up to polylogarithmic factors, to the simpler \emph{near neighbor} problem: Given an upper bound $r_+\geq r$, return a point within distance $cr_+$~\cite{DBLP:journals/toc/Har-PeledIM12}.
We refer to~\cite{andoni2018approximate} for more background on recent developments in approximate near neighbor search.

\paragraph{Models of computation}

Motivated by large-scale similarity search applications we consider models of hardware aimed at massive data sets.
The \emph{external memory model}~\cite{aggarwal} abstracts modern, block-oriented storage where memory consists of blocks each capable of holding $B$ data items. 
The cost of an algorithm or data structure is measured in terms of the number of block accesses, referred to as \emph{I/Os}.
When considering the $k$-NN problem we let $B$ denote the number of vectors that fit in a block.

Distributing similarity search onto many machines has also been considered as a way of scaling to large data sets~\cite{bahmani2012efficient, muja2014scalable, hu_output-optimal_2019}.
We can interpret a static, external memory data structure as an abstraction of a large, distributed system in which each server holds $B$ data items (and information associated with them).
In this context the parameter $B$ may be relatively large in comparison to $n$.
For example, to store the $n = 10^9$ DEEP1B vectors of dimension $d=96$, and associated raw data, we could imagine $s = 10^4$ servers each holding $B = 10^6$ data items (so that on average each item is replicated $10$ times to achieve redundancy).
The number of I/Os needed to answer a query then equals the number of servers that need to be involved when answering a $k$-NN query.

Most lower bounds for I/O efficient algorithms are shown under the assumption that data items are \emph{indivisible} in the sense that they are treated as atomic units, and that a block contains a well-defined set of at most $B$ data items.
The \emph{indexability model}~\cite{indexability1,indexability2}, introduced at PODS '97, formalized external memory data structures for queries that return a set of $k$ data items under the indivisibility assumption.
For a given data structure, the complexity of a query family $\mathcal{Q}$ is the smallest number of blocks that must be read to answer a query, in the worst case over all queries $q\in\mathcal{Q}$.
There does not need to be a constructive procedure for \emph{identifying} the correct blocks. 
In particular, the nearest neighbor problem ($k=1$) is trivial in the indexability model since the block containing the answer to the query is given for free; the search aspect is completely removed from consideration, and the algorithm would return the block in one I/O.
Though the original indexability model does not accommodate notions of approximation, it can be naturally extended to the setting where there is a set of at least $k$ elements that are valid answers to a query and we are required to return any $k$ of them.

%We present our results in Section~\ref{sec:results}, but first discuss related work.

%\noindent\textbf{Lower bounds and the indexability model:} Our lower bounds will be in the stronger indexability model proposed by Hellerstein et al.~\cite{indexability1}, and will therefore apply to any external memory data structure. In this model, the ``search'' aspect of our problem is given for free---i.e., upon receiving the query $q$, an oracle returns to the data structure the $k$ nearest neighbors of $q$. The data structure is now required to go to disk and ``fetch'' either those $k$ nearest neighbors, or any set of $c$-approximate $k$ nearest neighbors. The best that the data structure can do is fetch them in$\lceil k/B \rceil$ I/Os, but 

\subsection{Our results}\label{sec:results}

The complexity of general \cAkNN~queries in the I/O model lies between two extremes:
\begin{itemize}
\item There exist $\lceil k/B \rceil$ blocks (or servers) that contains a set of $k$ valid answers to the query, and
\item No block (or server) contains more than one valid query answer, so $k$ block reads are needed.
\end{itemize}
Since we do not care about constant factors we can, for simplicity, assume that $k < n/2$, since otherwise a trivial brute-force algorithm that reads all points is optimal within a factor of~2.
We give several upper and lower bounds for \cAkNN that suggest a dichotomy for polynomial space data structures.
For various choices of metric space, number of dimensions, and approximation factor we see that it is \emph{either} possible to achieve $O(k/B)$ I/Os, \emph{or} $\Omega(k)$ I/Os are provably required.

Our results are summarized in Table~\ref{table:results}. For our lower bounds, the dimension column represents the minimum (asymptotic) number of dimensions required for the bounds to hold.
Since it is possible to decrease the I/O complexity in lower dimensions~\cite{streppel2011approximate}, a condition on the number of dimensions is needed.
Our upper bounds do not depend on the number of dimensions, except indirectly through the definition of $B$ as the number of vectors in a block.
We stress that the I/O upper bounds for an indexing scheme do not imply the existence of a data structure with the same guarantees --- for a data structure in the I/O model we would expect an additional search cost. Our main theorems are:
%For simplicity we assume $k < n/2$. \rasmus{Check if we need $k<n^{1-\epsilon}$}

%\rasmus{State theorems here?}

\begin{table}[t]
%\caption{Table of our results}\label{table}
\begin{center}
\tabulinesep=1.5mm
\begin{tabu}{c|c|c|c} %to 0.8\textwidth { | X[c] | X[c] | X[c] | X[c] | X[c] | }
%\hline
     %{\bf Reference} & 
     {\bf Metric} & {\bf Dimension} & \begin{tabular}{c}{\bf Approximation}\\ {\bf factor}\end{tabular} & {\bf I/O bound} \\
        \hline
        \hline
        %Theorem~\ref{xxx} & 
        any & any & any & $\ge \lceil k/B \rceil$\\
        \hline
        %Theorem~\ref{xxx} & 
        $\ell_\infty$ & $\Omega(k \log(n/k))$ & $3-\epsilon$ & $\Omega(k)$ \\
         \hline
%         \vspace{2mm}Euclidean & \vspace{2mm}$\Omega(k \log n)$ & $1+\epsilon/k$ approx. in $<k$ I/Os requires $\Omega(N^{2}/kB^{2})$ blocks & \vspace{2mm}? \\
%         \hline
        %Theorem~\ref{xxx} & 
        Hamming & $\Omega(k \log (n/k))$ & $1+\tfrac{1}{4k}$ & $\Omega(k)$ \\
         \hline
        %Theorem~\ref{xxx} & 
        any & any & $3$ & $\le\lceil k/B \rceil$ \\
         \hline
        %Theorem~\ref{xxx} & 
        Hamming & any & $1+\epsilon$ & $\le\lceil k/B \rceil$ \\
         \hline
\end{tabu}
\end{center}
\caption{Our I/O lower bounds (first three rows) and upper bounds (last two rows) on the \cAkNN problem for data structures using polynomial space, where $\epsilon > 0$ can be any constant. The lower bound for $\ell_\infty$ holds assuming $k < n^{1-\epsilon}$.}\label{table:results}
\end{table}

\begin{theorem}\label{linfinitylower}(L-infinity metric lower bound)  Consider an indexing scheme for \cAkNN in $d$-dimensional $\ell_\infty$ space with $c<3$  and
  with worst case query time $\lceil k/\alpha \rceil$ I/Os, where $1 \leq \alpha \leq B$.
  For sufficiently large $d = O(k \log(n/k))$, the indexing scheme must use $\Omega \left(\left(\frac{1}{2e}\sqrt{\frac{n \alpha}{k B}} \right)^{\alpha}\right)$ blocks of space.
  \end{theorem}

\begin{theorem}\label{generalupper}(General metric indexing scheme) Given a set $P$ of $n$ points in any metric space, there exists a $3$-approximate indexing scheme that uses $n\lceil k/B \rceil$ blocks of space (where $B$ is the block size) and returns  $3$-approximate, $k$ nearest neighbors in optimal $\lceil k/B \rceil$ I/Os.
\end{theorem}
  
 \begin{theorem}\label{hamminglower}(Hamming metric lower bound)  Consider an indexing scheme for \cAkNN in $d$-dimensional Hamming space with $c<1+\frac{1}{4k}$ and 
  with worst case query time $\lceil k/\alpha \rceil$ I/Os, where $1 \leq \alpha \leq B$.
  For sufficiently large $d = O(k \log(n/k))$, the indexing scheme must use $\Omega \left(\left(\frac{1}{2e}\sqrt{\frac{n \alpha}{k B}} \right)^{\alpha}\right)$ blocks of space.
  \end{theorem}

 \begin{theorem}\label{hammingupper}(Hamming indexing scheme) Given a set $P$ of $n$ points in $d$-dimensional Hamming space and a constant $c>1$, there exists a $c$-approximate indexing scheme that uses $O(n ^{O(1)}d\lceil k/B \rceil)$ blocks of space and returns  $c$-approximate, $k$ nearest neighbors in optimal $\lceil k/B \rceil$ I/Os. 
\end{theorem}

\paragraph{Discussion.}
We are not aware of any studies of $c$-approximate, $k$-nearest neighbors in the I/O model.
However, our lower bounds suggest that such data structures cannot simultaneously have good output sensitivity (I/O complexity in terms of $k/B$) and low space usage.
The currently best data structures for high-dimensional, $c$-approximate $r$-near neighbor (return $1$ point within known distance~$r$) have query time lower bounded by $(n/B)^{f(c)}$ I/Os, for some non-increasing function $f$ with $f(1)=1$ that depends on the metric.
Depending on the choice of $c$ and $k$ this ``search cost'' may be smaller or larger than the cost of reporting the $c$-approximate $k$-nearest neighbors in the indexability model.
In general, we would expect the reporting cost to dominate only when $c$ and $k$ are sufficiently large.

Our lower bounds assume that the number of dimensions is sufficiently large.
An attempt to bypass the lower bounds would thus be to utilize some kind of \emph{dimension reduction}, such as the Johnson-Lindenstrauss lemma.
Unfortunately, this will not work in the settings of Theorems~\ref{linfinitylower} and~\ref{hamminglower}.
This is because dimension reduction is not possible for $\ell_\infty$~\cite{matouvsek1996distortion}, and because dimension reduction for Hamming space with distortion $c$ requires $\Omega(\log(n)/(c-1)^2)$ dimensions~\cite{larsen2017optimality}, which is $\Omega(k^2)$.

\subsection{Related work}\label{sec:related}

\paragraph{Lower bounds on nearest neighbors in restricted models}

A well-known work of Berchtold et al.~\cite{Berchtold1997ACM} analyzes the performance of certain types of nearest neighbor indexes on \emph{random} data sets.
More recently, Pestov and Stojmirovi{\'c}~\cite{pestov2006indexing} and Pestov~\cite{pestov2013lower} showed lower bounds for high-dimensional similarity search in a model of data structures that encompasses many tree-based indexing methods.
These results do not consider approximation, and their algorithmic models do not encompass modern algorithmic approaches to approximate similarity search such as locality-sensitive hashing.

\paragraph{Data structure lower bounds based on indivisibility}

There is a rich literature, starting with the seminal paper of Aggarwal and Vitter~\cite{aggarwal}, giving lower bounds on I/O-efficiency under an indivisibility assumption.
Such results in the context of data structures are known for dynamic dictionaries~\cite{wei2009dynamic}, planar point enclosure~\cite{arge2009optimal}, range sampling~\cite{hu2014independent}, and many variants of orthogonal range reporting~\cite{afshani2012improved, afshani2009orthogonal, indexability1, larsen2013near, yi2009dynamic}.
Below we elaborate on the most closely related works on orthogonal range queries.

To our best knowledge the high-dimensional $k$-NN problem has not been \emph{explicitly} studied in the indexability model~\cite{indexability1}.
However, in $n$-dimensional Hamming space it is straightforward to use the \emph{$k$-set workload} technique of~\cite{indexability1} to show that even obtaining $k-1$ I/Os is not possible (for $k\ll n$) unless the indexing scheme uses quadratic space.
Our lower bound technique is a generalization of the $k$-set workload that allows us to deal with approximation as well as space usage larger than quadratic.
It also allows us to show lower bounds all the way down to $O(k \log n)$ dimensions, as opposed to $n$ dimensions.

\paragraph{Orthogonal range queries}

\emph{Orthogonal range reporting} in $d$ dimensions asks to report all points in $P$ lying inside a query range $\mathcal{Q}$ that is a cross product of intervals.
% This problem was studied already in~\cite{indexability1}.
% Theorem 6.1: Redundancy > (log B/log A)^{d-1}, where A is access overhead
Note that $k$-NN in the $\ell_\infty$ metric is the special case of orthogonal range reporting where all intervals have the same length.
Hellerstein et al.~\cite{indexability1} showed that in order to answer orthogonal range reporting queries in $O(k/B^{1-\epsilon})$ I/Os, for some $\epsilon \in (0,1)$, the data structure needs to use space $\Omega(n / \epsilon^{d-1})$.
In particular, in dimension $d = \omega(\log n)$ a polynomial space data structure needs $k/B^{o(1)}$ I/Os.

\emph{Approximate} $d$-dimensional range reporting has been studied in the I/O model: Streppel and Yi~\cite{streppel2011approximate} show that for a query rectangle $q$ and constant $\epsilon > 0$, allowing the data structure to report points at distance up to $\epsilon \cdot\text{diam}(q)$ from $q$ makes it possible to report $k$ points in $2^{O(d)} + O(k/B)$ I/Os (plus a logarithmic search cost, which does not apply to the indexability model).
To our best knowledge, no lower bounds were known for approximate range queries before our work.

%A related problem is range reporting: given a data set $P$ in $\mathbb{R}^{d}$, construct a data structure which, given a query range $Q$, returns all points in $P$ contained in $Q$. The most extensively studied case is orthogonal range reporting, where the range $Q$ is a cross product of intervals. Another case is when $Q$ is a ball. In the case of a ball, the problem is similar to nearest neighbor search, as the points inside the ball will be the near neighbors of the center of the ball, say $q$. In the case of orthogonal range reporting, if the range is the cross product of equal-length intervals, say length $\ell$, then it can be seen as similar to a nearest neighbor problem in the $\ell_{\infty}$ metric.
%\riko{perhaps say something about 'near neighbor'}
%\mayank{The next two paragraphs introduce the approximate version}

\paragraph{Lower bounds based on computational assumptions}

Unconditional lower bounds for \cAkNN in the cell probe model~\cite{Barkol:2002:TLB:779034.779041, borodin1999lower, Chakrabarti:2004:ORC:1032645.1033203, panigrahy2010lower} only match upper bounds in the regime where space is very large.
To better understand the complexity for, say, sub-quadratic space usage, a possibility is to base lower bounds on computational assumptions such as the Strong Exponential Time Hypothesis (SETH), or the weaker Orthogonal Vectors Hypothesis (OVH).
Recently, Rubinstein~\cite{DBLP:conf/stoc/Rubinstein18} showed that under either of these hypotheses, for each constant $\delta > 0$, achieving an approximation factor of $c=1+o(1)$ is not possible for a data structure with polynomial space and construction time unless the query time exceeds $n^{1-\delta}$.
Already in 2001 Indyk showed that in the $\ell_\infty$ metric, \cAkNN with approximation factor $c<3$ is similarly hard~\cite{indyk2001approximate}.
(Though Indyk links the lower bound to a different problem, it can be easily checked that the same conclusion follows from the more recent SETH and OVH assumptions.)

\paragraph{Upper bounds}

The I/O complexity of the near neighbor problem was studied by Gionis et al.~\cite{DBLP:conf/vldb/GionisIM99}, focusing on Hamming space. (Their approach extends to other spaces that have good locality-sensitive hash functions.)
For every approximation factor $c>1$, they show that $O((n/B)^{1/c})$ I/Os suffices to retrieve \emph{one} near neighbor, using a data structure of size subquadratic in $n$.
It seems that the same algorithm can be adapted to return $k$ near neighbors at an additional cost of $O(k/B)$ I/Os.
Tao et al.~\cite{DBLP:journals/tods/TaoYSK10} extended these results to handle nearest neighbor queries, but they do not consider the case of $k$ nearest neighbors.

For the Euclidean and Hamming metrics with constant approximation factor $c>1$ it is known how to get $n^{o(1)}$ query time for the near neighbor problem ($k=1$) with a polynomial space data structure (see~\cite{andoni2017} and its references).
For $c>\sqrt{3}$ Kapralov~\cite{kapralov2015smooth} even showed how to achieve this by a \emph{single} probe to the data structure, returning a pointer to the result.

\paragraph{Organization:} The rest of this paper is organized as follows. In Section $2$ we develop notation and proceed in Section 2.1 to extend the indexability result of \cite{hellerstein1995} to approximate indexing schemes. Section $3$ describes the indexing scheme promised in Theorem~\ref{generalupper}, followed by Section $4$ that proves lower bounds in the L-infinity metric (Theorem~\ref{linfinitylower}). Section $5$ contains our results on the Hamming metric (Theorem~\ref{hamminglower} and Theorem~\ref{hammingupper}), followed by conclusion and open problems in Section $6$.

\section{Preliminaries}\label{sec:preliminaries}

The external memory %, or DAM (disk access machine) 
model of computation (due to Aggarwal and Vitter \cite{aggarwal}) has a main memory of size $M$ and an infinite external memory, both realized as arrays.
Data is stored in external memory, and is transferred to/from main memory (where computation happens) in I/Os, or \emph{block transfers}, where a block holds $B$ data items. 
Computation in main memory is free --- the cost of an algorithm is the number of I/Os it performs.

We will be using the following definitions from Hellerstein et al.~\cite{indexability1,indexability2}.
For brevity we will refer to a subset of $I$ of size $b$ as a \emph{$b$-subset}.
\begin{definition}[Workload]
A workload $W = (D,I,\mathcal{Q})$ consists of a non-empty set $D$ (the \textit{domain}), a nonempty finite set $I \subseteq D$ (the \textit{instance}, whose size we denote by $n$), and a set $\mathcal{Q}$ of subsets of $I$ (the \textit{query set}).
\end{definition}

\begin{definition}[Indexing Scheme]
An indexing scheme $S = (W,\mathcal{B})$ consists of a workload $W = (D,I,\mathcal{Q})$ and a set $\mathcal{B}$ of $B$-subsets of $I$, where $B$ is the \emph{block size} of the indexing scheme.
%such that $B$ is some positive integer and $\mathcal{B}$ \emph{covers} $I$, meaning that every query in $\mathcal{Q}$ is a subset of .
\end{definition}

% Basically, an indexing scheme is a way to lay out the data on the disk in blocks so as to answer queries efficiently. In the indexability model, once the query $q \in \mathcal{Q}$ arrives, there is an oracle that tells the algorithm exactly which blocks on disk contain the elements in $q$ for free (here we use $q$ both for the query and for its ``answers'', as we don't have to search for them). Ideally, we want indexing schemes that are space-efficient and do not perform too many I/Os to bring in these $|q|$ elements from memory. This is captured by the following two performance parameters.

% \begin{definition}[Redundancy]
% The redundancy $r(x)$ of a data item $x \in I$ is defined as the number of blocks $b$ in the block set $\mathcal{B}$ that contain $x$:

% $$ r(x) = \vert \{ b \in \mathcal{B} : x \in b \} \vert.$$

% The redundancy $r$ of the indexing scheme $S$ is defined as the average of $r(x)$ over all objects $x$, i.e., $r = \frac{1}{N}\sum_{x \in I}r(x)$. If the space used by the indexing scheme is $S$ blocks, $S = r N/B$.

% \end{definition}

% The redundancy ranges from $1$ (with $S= N/B$ blocks, the minimum required to store the data), to ${N \choose B}/ (N/B)$ (corresponding to $S = {N \choose B}$, where every $B$-combination of the input data is written out to a block).

\begin{definition}[Cover set]
 A cover set $C_Q \subseteq \mathcal{B}$ for a query $Q \in \mathcal{Q}$ is a minimum-size subset of the blocks such that $Q \subseteq \bigcup\limits_{b \in C_Q} b$.
\end{definition}

We will assume that the blocks are chosen such that $C_Q$ exists for every query $Q\in \mathcal{Q}$.

% \begin{definition}
% The access overhead $A(q)$ of a query $q \in \mathcal{Q}$ is defined as:

% $$ A(q) = \frac{C_q}{ \lceil |q|/B \rceil} .$$
% \end{definition}

% \begin{definition}[Access Overhead]
% The access overhead $A$ for an indexing scheme $S$ is defined as  $A = \max\limits_{q \in \mathcal{Q}} A(\mathcal{Q}) $.
% \end{definition}

% Given a query $q$, any indexing scheme must make at least$\lceil |q|/B \rceil$ I/Os to answer $q$. Intuitively, if the access overhead of an indexing scheme is $A$, then for some query $q$, the scheme requires $A \lceil |q|/B \rceil$ I/Os to answer $q$. 

\begin{definition}[$\lambda$-Set Workload]
The $\lambda$-set workload is a workload with instance $I = \{1, \cdots, n\}$ whose
query set $\mathcal{Q}$ is the set of all $\lambda$-subsets of the instance.
\end{definition}

While~\cite{indexability1,indexability2} measure performance in terms of \emph{redundancy} and \emph{access overhead}, we find it more natural to define performance in the same way as the I/O model.
The \emph{space usage} of an indexing scheme is the number of blocks, $|\mathcal{B}|$.
The \emph{query time} of $Q \in \mathcal{Q}$ is the size of a cover set,~$|C_Q|$.
Observe that the time of a query of size $\lambda$ can range from $\lceil \lambda/B \rceil$ I/Os to $\lambda$ I/Os, depending on how many blocks are needed to cover it.

All static data structures we know of in the external memory model with a given space usage and query cost translate directly to an indexing scheme with the same, or better, performance.
This is because these data structures store, or can be adapted to store, $O(B)$ vectors explicitly in each block, such that the set of result vectors is explicitly present in the blocks that are read when answering a query.
This means that any lower bounds in the indexability model strongly suggest lower bounds also in the external memory model.
%We are only aware of one external memory data structure that does not seem to have corresponding indexing scheme, namely the dictionary of Iacono and Patrascu. Hmm... this is dynamic.
%Note that the indexing scheme knows $\lambda$; in general, an indexing scheme is aware of the exact workload.

% \begin{theorem}\label{setworkloadtheorem}\cite[Theorem~7.2]{indexability1}
% For $\lambda$-set workload $K_{n,\lambda}(I,\mathcal{Q}), B \geq \lambda$, any indexing scheme with redundancy

% $$ 
% r < \frac{n - \lambda + 1}{(\lambda -1)(B - 1)}
% $$

% has worst possible access overhead $A = \lambda$.
% \end{theorem}

% The above theorem says that if an indexing scheme is to avoid the worst possible access overhead, it must use space at least quadratic in $N/B$. 

%\subsection{$c$ $k$-Nearest Neighbor Problem}
%
%\begin{definition} 
%\textbf{$k$-Nearest Neighbor Problem}: Given a set $P=\{p_{1}, \cdots, p_{n}\}$ of $n$ points, where $p_{i}$ belongs to a metric space $M = (X, d)$ , and a query $q$, return a set $K = \{t_1,  \cdots ,t_k \}$ of $k$ points from $P$ such that for any $p \in P \setminus K$ then $d(q, t_i) < d(q, p)$, for all $1 \leq i \leq k$
%\end{definition}

\subsection{Approximate indexing schemes}

\begin{definition}[$c$-approximate $k$-nearest neighbors problem, \cAkNN] Let $M$ be a metric space with distance function $d(\cdot,\cdot)$, and $c>1$ be a constant. Given $P \subset M$, $\vert P \vert= n$, and a positive integer $k \leq n$, construct a data structure that upon receiving a query $q$, returns $Q \subset P$, $Q=\{q^{1},\cdots,q^{k}\}$, such that for all $q^{i} \in Q$, $d(q,q^{i}) \leq c\, d(q, p^{k}_{q})$, where $p^{k}_{q} \in P$ is the $k$th nearest neighbor of $q$ in $P$.
\end{definition}

We will consider indexing schemes for \cAkNN that can depend on the parameters $c$ and $k$, i.e., the parameters are known when the data structure is constructed.
The set of points to be stored is denoted by $P$.
Observe that for \cAkNN, there can be at most $\binom{n}{k}$ distinct query answers ($\lambda=k$ in the $\lambda$-set workload above). Intuitively, the more queries there are in a workload, the higher is the space needed by the data structure. 
But in the \cAkNN problem, for a query $q$ there is no one ``right'' answer, as any set of $k$ approximate near neighbors form a valid output. 
Thus we find that the definitions in the original indexability model need to be extended to capture approximate data structures
It turns out that even though approximate near neighbors are allowed, it is possible to ensure that at least half of the $k$ returned points are among the $k$ nearest to the query point.
This motivates the following definition:

\begin{definition}[Relaxed $\lambda$-set workload]
The relaxed $\lambda$-set workload  is a workload whose instance is $\{1, \cdots, n\}$. The query set is the set of all $\binom{n}{\lambda}$ subsets of size $\lambda$. Given a query $q$ corresponding to such a $\lambda$-subset $Q$, the indexing scheme must report $\lambda$ elements, at least half of which must come from $Q$.
\end{definition}

%In the above relaxed set workload we only require the indexing scheme to report half of the required exact query output. 
We next show a space-I/O tradeoff for the relaxed workload. 

\begin{lemma}[Relaxed set workload tradeoff]\label{relaxedtradeoff}
  Any indexing scheme for the relaxed $\lambda$ set workload with a space usage of $s$ blocks and a query time of $t$ block accesses
  must have
  \[
   t \geq \lambda\log_s \left(\frac{1}{2e}\sqrt{\frac{n}{tB}}\right)
  \]
\end{lemma}
\begin{proof}
 
 Recall that by the relaxation, the indexing only needs to find a subset~$\hat Q \subset Q$ of size $|\hat Q| =\lambda/2$ in the index. In other words, with each subset $\hat Q$ that can be retrieved by the index, the algorithm can add any $\lambda/2$ elements and arrive at a valid query~$Q$.

  We upper bound the total number of distinct $\lambda$-sets reported by an indexing scheme using $s$ space and $t$ query I/Os as follows:
  \begin{itemize}
  \item choose the set of $t$ blocks to retrieve from the index ($\binom{s}{t}$ choices),
  \item choose the $\lambda/2$ elements to use from these at most $tB$ distinct elements (at most $\binom{{tB}}{{\lambda/2}}$ choices),
  \item choose $\lambda/2$ arbitrary \emph{other} elements (at most $\binom{n}{{\lambda/2}}$ choices).
  \end{itemize}
  
  The total number of such combinations should be at least $\binom{n}{\lambda}$, which is the possible set of queries, or $\lambda$-subsets.
  Using the inequalities $\left(\frac{n}{k}\right)^{k} \leq \binom{n}{k} \leq \left(\frac{en}{k}\right)^{k} \leq n^{k}$ this gives us:

\begin{eqnarray}
\binom{s}{t} \binom{tB}{\lambda/2} \binom{n}{\lambda/2} &\geq& \binom{n}{\lambda}  \label{eqn:tradeoff} \\
\Rightarrow s^t \left(\frac{tBe}{\lambda/2}\right)^{\lambda/2} \left(\frac{ne}{\lambda/2}\right)^{\lambda/2}
  &\geq& \left(\frac{n}{\lambda}\right)^\lambda \notag\\
\Rightarrow s^{2t/\lambda} \left(\frac{tBe}{\lambda/2}\right) \left( \frac{ne}{\lambda/2}\right) &\geq& \left(\frac{n}{\lambda}\right)^2  \notag \\
\Rightarrow s^{2t/\lambda} & \geq  & %\frac{n^2}{{\lambda}^2}\frac{{\lambda}^2}{4e^2tBn} &=& 
\frac{n}{4e^2tB} \notag \\ 
\Rightarrow t &\geq& \lambda\log_s \left(\frac{1}{2e}\sqrt{\frac{n}{tB}}\right). \notag
\end{eqnarray}

% \begin{equation}\label{eqn:tradeoff}
% \binom{s}{t} \binom{tB}{\lambda/2} \binom{n}{\lambda/2} \geq \binom{n}{\lambda}
% \end{equation}

% \[
%   s^t \left(\frac{tBe}{\lambda/2}\right)^{\lambda/2} \left(\frac{ne}{\lambda/2}\right)^{\lambda/2}
%   \geq \left(\frac{n}{\lambda}\right)^\lambda
% \]

% \[
% s^{2t/\lambda} \left(\frac{tBe}{\lambda/2}\right) \left( \frac{ne}{\lambda/2}\right) \geq \left(\frac{n}{\lambda}\right)^2
% \]
% \[
% s^{2t/\lambda} \geq \frac{n^2}{{\lambda}^2}\frac{{\lambda}^2}{4e^2tBn} = \frac{n}{4e^2tB} 
% \]
% \[
%   t>\lambda\log_s \left(\frac{1}{2e}\sqrt{\frac{n}{tB}}\right)
% \]
\end{proof}

One may be interested in a query time of $\lceil \lambda/\alpha \rceil$ I/Os, where $1 \leq \alpha \leq B$. 

\begin{corollary}\label{relaxedtradeoffpolynomial}
Any indexing scheme for the relaxed $\lambda$-set workload with worst case query time of $\lceil \lambda/\alpha \rceil$ I/Os, where $1 \leq \alpha \leq B$ must use $\Omega \left(\left(\frac{1}{2e}\sqrt{\frac{n \alpha}{\lambda B}} \right)^{\alpha}\right)$ blocks of space.
\end{corollary}

This lower bound is essentially tight: Achieving a query time of $\lceil \lambda/\alpha \rceil$ I/Os using $O(\binom{n}{\alpha})$ space is easy by just preprocessing all $\alpha$ combinations of the $n$ items. By doing an analogous calculation for the exact workload, we are able to give the following tradeoff for the standard $\lambda$-set workload.

%\rasmus{In other places we phrase this as an I/O lower bound for fixed polynomial space $n^c$. Do the same here?}
%\mayank{Tried to do that above. Please take a look.}
\begin{lemma}[$\lambda$-set workload tradeoff]\label{relaxedAlphaTO}
Any indexing scheme for the $\lambda$-set workload with worst case query time of $\lceil \lambda/\alpha \rceil$ I/Os, where $1 \leq \alpha \leq B$ must use $\Omega \left(\left(\frac{n \alpha}{e \lambda B} \right)^{\alpha}\right)$ blocks of space.
\end{lemma}
%\riko{I guess there is a gap of a factor of 2 in the exponent of the space, right?}

Lemma~\ref{relaxedAlphaTO} generalizes Theorem 7.2 in \cite{indexability1} which considered the case of $\alpha =1$.

\section{A 3-approximation indexing scheme for general metrics}\label{sec:3apx}

We prove Theorem~\ref{generalupper} in this section, that asserts that (a relatively simple) indexing scheme provides a $3$-approximation for the $k$-NN problem in any metric space. Note that we are only presenting an indexing scheme as opposed to a data structure; i.e., we  assume that once the query is given, an oracle provides the smallest set of blocks that contain a valid answer to the query.
%the list of the exact $k$ nearest neighbors. One must then go and fetch blocks containing $3$-approximate $k$ nearest neighbors.

Let $P$ be a set of $n$ points in a metric space. Consider the indexing scheme $I$ that consists of the set of $\ell=n\lceil k/B \rceil$ blocks $\mathcal{B} = \{ b_1, b_2, ... , b_\ell \}$ where we store $p_i \in P$ and the set $P_{i} = \{ p_{i,1}, p_{i,2}, ..., p_{i,k}\}$ of the $k$-nearest neighbors of the point $p_i$ (including itself). 
This requires $\lceil k/B \rceil$ blocks per element of $P$, as claimed.
%
%Upon learning the nearest neighbor $p^{1}_{q}$
%\riko{did we introduce this notation? Do we have to change to $p^{(1)}$ to avoid confusion with exponentiation?}
For a query point $q$ let $p_{i^*}$ be the nearest neighbor of $q$.
The oracle then returns the set $P_{i^*}$, i.e., the $k$ nearest neighbors of $q$'s nearest neighbor, using $\lceil k/B \rceil$ I/Os.

\begin{theorem}
Let $\{p^{q}_{1}, \cdots,p^{q}_{k}\}$ be the set of the exact $k$ nearest neighbors of a query $q$, and let $\{p_{i^*,1},\cdots p_{i^*,k}\}$ be the set of points returned by the indexing scheme described above. 
Then for any $1 \leq j \leq k$, $d(q,p_{i^*,j}) \leq 3 d(q,p^{q}_{k})$. 
That is, all the returned points are within a factor $3$ of the distance of the query to its $k$th nearest neighbor. 
%Furthermore, this indexing scheme reports $k$ approximate near neighbors in $\lceil k/B \rceil$ I/Os.
\end{theorem}

\begin{proof}
%Note that $\{q^{1},\cdots q^{k}\}$ is exactly the set of the $k$ nearest neighbors of $p^{1}_{q}$. 
The proof is a case analysis. Let $D$ be the smallest ball centered at $p^{q}_{1} := p_{i^*}$ that contains $p_{i^*,k}$ (the ``$k$ nearest neighbor disk'' of $p^{i^*}$). The cases are:
\begin{enumerate}
    \item Both $q$ and $p^{q}_{k}$ are outside $D$. For each $j$, we have that 
    \begin{align*}
    d(q, p_{i^*,j}) &\leq d(p_{i^*,j},p^{q}_{1}) + d(p^{q}_{1},q) \ \ \ \text{(triangle inequality)}\\
                &\leq d(p_{i^*,j}, p^{q}_{1}) + d(q, p^{q}_{k}) \ \ \ \text{($p^{q}_{k}$ is farther from $q$ than $p^{1}_{q}$)} \\
                &\leq d(q,p^{q}_{1}) + d(q, p^{q}_{k}) \ \ \ \text{($q$ is outside $D$)}\\
                &\leq 2\,d(q, p^{q}_{k}) \enspace . 
\end{align*}
\item $q$ is inside $D$, and $p^{q}_{k}$ is outside $D$. For each $j$, we have that 
  \begin{align*}
    d(q, p_{i^*,j}) &\leq d(p_{i^*,j},p^{q}_{1}) + d(p^{q}_{1},q) \ \ \ \text{(triangle inequality)}\\
                & \leq d(p_{i^*,j},p^{q}_{1}) + d(q, p^{q}_{k}) \ \ \ \text{($p^{q}_{k}$ is farther from $q$ than $p^{q}_{1}$)}\\
                &\leq d(p^{q}_{k}, p^{q}_{1}) + d(q, p^{q}_{k}) \ \ \ \text{($p^{q}_{k}$ is outside $D$)}  \\
                &\leq d (p^{q}_{k}, q) + d(q,p^{q}_{1}) + d(q, p^{q}_{k}) \ \ \ \text{(triangle inequality)}\\
                &\leq 3\, d(q, p^{q}_{k}) \enspace .
\end{align*}
\item Both $q$ and $p^{q}_{k}$ are inside $D$. In this case, we first claim that it is sufficient to consider that for some $i \notin \{1,k\}$, $p^{q}_{i}$ is outside $D$. If that was not the case, we have reported all $\{p^{q}_{i}\}$, i.e., the exact $k$ nearest neighbors of $q$. For each $j$ we now have that
 \begin{align*}
        d(q, p_{i^*,j}) &\leq d(q,p^{q}_{1}) + d(p^{q}_{1}, p_{i^*,j}) \ \ \ \text{(triangle inequality)} \\
                    &\leq d(q, p^{q}_{k} ) + d(p^{q}_{1} , p_{i^*,j}) \ \ \ \text{($ p^{q}_{k}$ is farther from $q$ than  $p^{q}_{1}$)} \\
                    &\leq d(q, p^{q}_{k} ) + d(p^{q}_{1} ,p^{q}_{i}) \ \ \ \text{($ p^{q}_{i}$ is outside $D$)} \\
                    &\leq d(q, p^{q}_{k} ) + d(p^{q}_{1} ,q) + d(q,p^{q}_{i}) \ \ \ \text{(triangle inequality)} \\
                    & \leq 3\, d(q, p^{q}_{k}) \enspace .
    \end{align*}
\item $q$ is outside $D$, and $p_{k}^{q}$ is inside $D$. In this case we have that 
    \begin{align*}
                  d(q, p_{i^*,j}) &\leq d(p_{i^*,j},p^{q}_{1}) + d(p^{q}_{1},q) \ \ \ \text{(triangle inequality)}\\
                  & \leq d(q,p^{q}_{1}) + d(p^{q}_{1},q) \ \ \text{($q$ is outside $D$; $p_{i^*,j}$ is inside)}\\
                    & \leq 2d(q,p_{k}^{q}) \ \ \ \text{($p^{q}_{k}$ is farther from $q$ than $p^{q}_{1}$)}\\
                \end{align*}
\end{enumerate}
\end{proof}

\paragraph{Tightness of the analysis:} 
Next, we give an example of $n=k+1$ points where the indexing scheme above does not achieve a better approximation factor than $3(1-\epsilon)$, for arbitrary $\epsilon>0$.
Consider the scenario when all points including the query lie on a line: 
$$q = 1/2,\quad p^{q}_{1}= 0, \quad p^{q}_{k}=1+\epsilon/2, \quad p^{q}_{k+1} = -1, \quad \text{and}$$
$$p^{q}_{i} = -\epsilon i/2k \text{ for all } i \notin \{1,k,k+1\} \enspace . $$
The disk $D$ contains all points but $p^{q}_{k}$, $p^{q}_{k+1}$ will be reported, and $d(q,p^{q}_{k+1}) = 3/2 > 3(1-\epsilon)(1/2+\epsilon/2) = 3(1-\epsilon)d(q,p_{k}^{q})$.

\section{Lower bound for L-infinity metric}

In this section we prove Theorems~\ref{linfinitylower} and~\ref{hamminglower}.

\subsection{Lower bound for the L-infinity metric}

We warm up with a proof of a lower bound in $n$ dimensions, a slight variant of a reduction by Indyk~\cite{indyk2001approximate}. 

\begin{lemma}\label{3apxLBhighdimSetWorkload}
  Consider the approximate $k$-NN problem in $\mathbb{R}^{n}$ with $\ell_\infty$ metric.
  There exists a set of~$n$ points $p_i$ for $i\in \{1,\ldots,n\}$ that has a $k$-set workload. Specifically,
  for every subset $I\subset [n]$ there exists a query point $q_I$ such that the $\ell_{\infty}$-distance is
  \[
    \|p_i-q_I\|_\infty =
    \left\{\begin{array}{lr}
        1/2 & \text{for } i \in I\\
        3/2 & \text{for } i \not\in I
        \end{array}\right. \;.
  \]
  Hence, the \cAkNN\ for $c<3$ in $n$-dimensional $\ell_\infty$  space leads to a $k$-set workload.
\end{lemma}
\begin{proof}
  The set consists of the $n$ unit vectors $p_i = e_i$ (where only the $i$th entry is 1 and all other entries are 0).
  For a set~$I\subset \{1,\dots,n\}$ the query vector is defined as
  \[
    q_I = 
    \left\{
      \begin{array}{lr}
        +1/2 & \text{for } i \in I\\
        -1/2 & \text{for } i \not\in I
      \end{array}\right. \;.
  \]
\end{proof}

To deterministically reduce the dimensionality of the space we use an expander and switch to relaxed $k$-set workloads.
Expanders were previously used for deterministic embeddings of Euclidean space into $\ell_1$ by Indyk~\cite{DBLP:conf/stoc/Indyk07}.
There is a vast literature on expanders and the results we are using are standard by now. 
For the sake of concreteness, we take the definitions and precise results almost literally from~\cite{RasmusAnnaExpander02}.
We define $(m, \delta,1/3)$-expander graphs and state some results concerning these graphs. For the rest of this paper we will assume $\delta$ to divide $1/3$, as this makes statements and proofs simpler. This will be without loss of
generality, as the statements we show do not change when rounding~$\delta$ down to the nearest such value.
Let $G = (U, V, E)$ be a bipartite graph with left vertex set $U$ , right vertex
set $V$, and edge set $E$.
We denote the set of neighbors of a set $S \subseteq U$ by
$\Gamma  (S) = \displaystyle\cup_{s\in S} \{v \mid (s, v) \in E\}$, and
use $\Gamma(x)$ as a shorthand for $\Gamma (\{x\})$, $x \in U$.
\begin{definition}[Definition 3 of \cite{RasmusAnnaExpander02}]
  A bipartite graph $G = (U, V, E)$ is $\delta$-regular if the degree of all
  nodes in $U$ is $\delta$.
  A bipartite $\delta$-regular graph $G = (U, V, E)$ is an $(m, \delta, 1/3)$-expander
  if for each $S \subseteq U$ with $|S| \leq m$ it holds that $|\Gamma (S)| \geq (1 - 1/3)\delta |S|$.
\end{definition}

\begin{lemma}[Corollary 5 in \cite{RasmusAnnaExpander02}]\label{lem:cor5Exp}
  For every constant $\epsilon > 0$ there exists an
  $(m, \delta, \epsilon)$-expander $G =(U, V, E)$ with $|U | = u$, $\delta = O(\log(2u/m))$ and $|V| = O(m \log(2u/m))$.
%  \item ($|U | = u$, $\delta = O(1)$ and $|V | = O(m (2u/m)^\alpha )$.)
%  \end{itemize}
\end{lemma}

Next, we discuss how to give an analogue of the hard query set in Lemma~\ref{3apxLBhighdimSetWorkload} with $O(k\log n)$ dimensions.
\begin{lemma}
  Let $n,k(n)$ be arbitrary integer parameters. %\riko{for the moment}
  Consider \cAkNN with $c<3$.
  There exists a set~$D$
  of $n$ points in dimension $O(k\log n)$ such that
  for any $I\subset [n]$ with $|I|=k$ there exists a query point $q_I$ such that
  the set $I'= \{p\in D \mid \|p-q_I\|_\infty \le 1/2 \}$ of potential answer points
  has $I\subset I'$ and $|I'|\leq 3k/2$.
\end{lemma}

\begin{proof}
  Fix $n$ and $k(n)$.
  Let $m=k(n)$ and let $G =(U, V, E)$ be an 
  $(m, \delta=O(\log(n/m)), 1/3)$-expander with $|U|=n$ and $|V|=d=O(k\delta)$.
  For concreteness we take $U=[n]$.

  Construct the set of points $p_1,\dots,p_n$ where
  \[
    (p_i)_j =
    \left\{
      \begin{array}{lr}
        1 & \text{if } j \in \Gamma(i)\\
        0  & \text{otherwise}
      \end{array}\right. \;.
  \]
  Define the query point for set $I$ with $|I|=k$ where
  \[
    (p_I)_j =
    \left\{
      \begin{array}{lr}
        +1/2 & \text{for } j \in \Gamma(I)\\
        -1/2 & \text{otherwise}
      \end{array}\right. \;.
  \]

  It is easy to see that $\|p_i - p_I\|_\infty = 1/2$ for all $i\in I$, so $p_I$ has at least $k$ neighbors at distance~$1/2$.
  It remains to show that this leads to a relaxed $k$-set workload, i.e.,
  that for any $q_I$ any set of 1/2-near points in the set has at least $k/2$ points in common with~$I$.
  Fix a subset~$I$ and consider the query point~$q_I$.
  Let $I' = \{ i \mid \|p_i - p_I\|_\infty\leq 1/2 \}$ be the indices of the points of with distance at most 1/2 to~$q_I$.
  Observe that $I\subseteq I'$ by construction.
  Let $I^*= I'\setminus I$ be the indices of ``unintended near points''.
  Observe that every point $p_i$ that does not fulfill $\Gamma(i)\subset \Gamma(I)$ has $\|p_i - p_I\|_\infty = 3/2$ and cannot be reported since $c<3$.
  Hence
  \[ I^* = \{ i \mid \Gamma(i)\subseteq \Gamma(I) \} \backslash I\;. \]
  Observe
\( |\Gamma(I)|\leq  \delta |I|\) and
\( |\Gamma(I \cup I^*)| \geq (1-1/3) \delta  |I \cup I^*| \), by the definition of~$G$ and 
\( |\Gamma(I \cup I^*)| =  |\Gamma(I)| \) by the definition of~$I^*$.
Combining this we get
\( \delta |I| \geq \tfrac{2}{3} \delta  |I \cup I^*| \),  
or \( |I| \geq \tfrac{2}{3} (|I|+|I^*|) \).
That is,
\( \frac{1/3}{2/3} |I| \geq |I^*|\),
as desired.
Hence we have a relaxed $k$-set workload.
\end{proof}

This means that we can apply Lemma~\ref{relaxedtradeoff} and Corollary~\ref{relaxedtradeoffpolynomial}:
\begin{corollary}(Theorem~\ref{linfinitylower})
  Any indexing scheme for
  \cAkNN in $O(k \log(n/k))$-dimensional $\ell_\infty$ space with $c<3$ 
  with worst case query time of $\lceil k/\alpha \rceil$ I/Os, where $1 \leq \alpha \leq B$ must use $\Omega \left(\left(\frac{1}{2e}\sqrt{\frac{n \alpha}{k B}} \right)^{\alpha}\right)$ blocks of space.
\end{corollary}

\subsection{Lower Bound for the Hamming metric}\label{sec:hamming_lower}

We now prove Theorem~\ref{hamminglower}, giving a lower bound on indexes on sets of vectors in Hamming space with approximation factor $c>1$.
This directly implies a lower bound in the $\ell_1$ metric, as well as a lower bound for $\ell_2$ with approximation factor $\sqrt{c}$.

\begin{lemma}\label{3apxLBhammingHighD}
  Consider the approximate $k$-NN problem in Hamming space of dimension~$d$.
  There exists a set of~$n$ points $p_i$, $i\in \{1,\ldots,n\}$ of dimension $d=n$
  that is a $k$-set workload, i.e,
  for every subset $I\subset [n]$ with $|I|=k$ there exists a query point $q_I$ such that the Hamming (or $\ell_1$) distance is:
  \[
    \|p_i-q_I\|_1 =
    \left\{
      \begin{array}{lr}
        k-1 & \text{for } i \in I\\
        k+1  & \text{for } i \not\in I
      \end{array}\right. \;.
  \]
  Hence, the \cAkNN\ for $c<1+2/(k-1)$ in $n$-dimensional Hamming space leads to a $k$-set workload.
\end{lemma}
\begin{proof}
  The set consists of the $n$ unit vectors $e_i$. The query vector for set~$I$ is the characteristic vector of~$I$. 
  It is easy to verify that the distances from the query to the vectors in $I$ is $k-1$, and to those not in $I$ is $k+1$. 
  Note that $(k+1)/(k-1) = 1 + 2/(k-1)$, which gives the bound on~$c$.
\end{proof}

\begin{lemma}\label{3apxLBhammingLowD}
  Consider the approximate $k$-NN problem in Hamming space of dimension~$d$.
  There exists a set of~$n$ points $p_i$, $i\in \{1,\ldots,n\}$, and $d=O(k\log(n/k))$
  that is a relaxed $k$-set workload, i.e,
  for every subset $I\subset [n]$ with \(|I|=k\) there exists a query point $q_I$ such that the Hamming ($\ell_1$)-distance is
  \[
    \|p_i-q_I\|_1 =
    \left\{
      \begin{array}{lr}
        \le d_I & \text{for } i \in I\\
        \ge d_I \left(1+\frac{1}{4(k-1)}\right) & \text{for } i \not\in I'
      \end{array}\right. \;
  \]
  for an $I' \supset I$ with $|I'| \le 1.5|I|$.
  Hence, $c<1+\frac{1}{4(k-1)}$ leads to a relaxed $k$-set workload for \cAkNN\ in $O(k\log n)$-dimensional Hamming space.
\end{lemma}

\begin{proof}
  Fix $n$ and $k(n)$.
  Let $G =(U, V, E)$ be an 
  $(m=k(n), \delta=O(\log n), 1/4)$-expander
  with $|U|=n$ and $|V|=d=O(k\log n)$, which exists by Lemma~\ref{lem:cor5Exp}.
  For concreteness we take $U=[n]$.

  Construct the set of points $p_1,\dots,p_n$ in the same way as earlier:
  \[
    (p_i)_h =
    \left\{
      \begin{array}{lr}
        1 & \text{for } h \in \Gamma(i)\\
        0  & \text{otherwise}
      \end{array}\right. \;.
  \]
  We define the query point $q_I$ for each set $I$ with $|I|=k$ where
  \[
    (q_I)_h =
    \left\{
      \begin{array}{lr}
        1 & \text{for } h \in \Gamma(I)\\
        0  & \text{otherwise}
      \end{array}\right. \;.
  \]
  In other words, the vectors are the characteristic vectors of the neighbor-sets in~$G$.

  We define \( d_j  :=\|p_j-q_I\|_1  \).
  For $i\in I$ we have \( d_i = |\Gamma(I)|-\delta \)
  and 
	\[d_i \geq k\delta(1-1/4) - \delta = \delta(k(1-1/4)-1)\enspace. \]
  For any $j$ (in particular \(j\notin I\)) we have \( d_j = \|p_j-q_I\|_1 = |\Gamma(I)|-\delta + 2|\Gamma(j)\setminus \Gamma(I)| \), leading to the definition $\gamma_j = |\Gamma(j)\setminus \Gamma(I)|$.
  We can set the distance threshold ratio such that the number of unintended neighbors is again at most $k/2$:
  
  We can set the distance threshold ratio to $1+1/(4 (k-1))$  such that the set of unintended near neighbors are the $j$ with $\gamma_j < \delta/4$:
  \[
    d_j/d_I = 1 + \gamma_j/d_I  \leq 1+1/(4(k-1)) \Rightarrow \gamma_j \leq d_I/(4 (k-1))
    \leq \delta/4
  \]
  To calculate the set of indices of unintended neighbors we define
  \[ I^* = \{ j \mid \gamma_j \leq \delta/4 \} \setminus I  \enspace . \]
	Then
  \[ |\Gamma(I \cup I^*)| \leq \delta  |I| + (\delta/4) |I^*| \]
  \[ |\Gamma(I \cup I^*)| \geq (1-1/4)\delta |I \cup I^*| = \tfrac{3}{4} \delta (|I|+| I^*|)\]
  leading to 
  \[ \delta  |I| + (\delta/4) |I^*| \geq \tfrac{3}{4} \delta (|I|+| I^*|) \]
  \[  \Rightarrow (\delta/4) |I^*| \geq \tfrac{3}{4} \delta | I^*| - |I| \delta /4  \]
  \[  \Rightarrow |I| \delta / 4 \geq \tfrac{3}{4} \delta | I^*| - (\delta/4) |I^*| \enspace . \]
  Hence
  \[ | I^*| \leq |I|  \frac{1/4}{\tfrac{3}{4}-1/4} = k/2 \enspace .
  \]
  This means that the described workload is a relaxed $k$-set workload. Applying Lemma~\ref{relaxedtradeoff} and Corollary~\ref{relaxedtradeoffpolynomial} now completes the proof of Theorem~\ref{hamminglower}.
\end{proof}

\section{Hamming metric indexing scheme}\label{sec:hamming_upper}

The general 3-approximate indexing scheme described in Section~\ref{sec:3apx} can be improved for specific metrics.
In this section we prove Theorem~\ref{hammingupper}. For any given approximation factor $c>1$ we wish to construct an indexing scheme that answers \cAkNN queries in $d$-dimensional Hamming space using polynomial space and with $\lceil k/B\rceil$ I/Os.
Our construction is an application of the dimension reduction technique of Kushilevitz et al.~\cite{kushilevitz2000efficient}.

For each $r\in\{1,\dots,d\}$ we create a data structure that handles the case where the $k$th closest point to $q$ is at distance $r$.
The data structure must report $k$ points that have distance at most $cr$ from~$q$.
The central idea of~\cite{kushilevitz2000efficient}  is to use a randomized mapping 
$$t: \{0,1\}^d\rightarrow \{0,1\}^D,$$ 
where $D=O(\log n)$, such that for each $q\in\{0,1\}^d$ with high probability for all $x,y\in P$:
\begin{equation}\label{eq:distinction}
d(q,x)\leq r \wedge d(q,y) > cr \implies  d(t(q),t(x)) < d(t(q),t(y)) \enspace .
\end{equation}
(We note that the required dimension $D$ grows as $c$ approaches $1$, hence we need to keep $c$ fixed.)
Consider the mapped multiset $t(P) = \{ t(x) \,|\, x\in P\}$ and create a data structure that for each $i\in \{0,1\}^D$ lists, for the $k$ nearest neighbors of $i$ in $t(P)$, the corresponding vectors in $P$ (breaking ties arbitrarily), using $\lceil k/B\rceil$ blocks.
If (\ref{eq:distinction}) holds then list $i=t(q)$ contains only $c$-approximate $k$-nearest neighbors of $q$.
To eliminate the error probability, choose $O(d)$ such random mappings and construct corresponding data structures: With high probability there will be no query $q\in\{0,1\}^d$ that does not have at least one data structure that returns a correct result.
If this fails for some $q$, start over from the beginning and choose new mappings.

The total space usage is $O(2^D d \lceil k/B\rceil)$, which is polynomial in $n$ and $d$, as desired.
Queries can be answered in $\lceil k/B\rceil$ I/Os since we are taking full advantage of the power of the indexability model: To answer a query $q$ it is necessary to know which mapping $t$ can be used to answer it correctly, and where in storage the blocks with index $t(q)$ reside.

Of course, we can also get an algorithm in the standard I/O model with a multiplicative query time overhead of $O(d)$ by querying \emph{all} repetitions and returning the $k$ closest points seen.

 \section{Conclusion and open problems}
 
We have shown that nontrivial lower bounds can be shown in the indexability model, even under approximation.
The main open problem that we leave is whether our hardness result for Hamming distance can be extended to approximation factor $c = 1+ \Omega(1)$, where the constant in $\Omega(1)$ is independent of $k$.
This would give an unconditional analogue of the recent conditional lower bound of Rubinstein~\cite{DBLP:conf/stoc/Rubinstein18}.

\bibliographystyle{plain}
\bibliography{knnindexability}

\end{document}